\documentclass[final, authoryear, 3p,12pt]{elsarticle}

\usepackage{amssymb}
\usepackage{amsthm}

\usepackage{mathrsfs}
\usepackage{float}
\usepackage{eufrak}
\usepackage{amsmath}
\usepackage{graphics}
\usepackage{amsfonts}
\usepackage{mathrsfs}
\usepackage{amscd}
\usepackage{color}
\usepackage{url}
\usepackage[plainpages=false,pdfpagelabels=true,colorlinks=true,citecolor=blue,hypertexnames=false]{hyperref}
\usepackage[vlined,ruled]{algorithm2e}

\usepackage{listings}
\journal{}

\newtheorem{theorem}{Theorem}[section]

\newtheorem{remark}[theorem]{Remark}
\newtheorem{define}[theorem]{Definition}
\newtheorem{example}[theorem]{Example}

\def\gf{{\mathbb{F}}}
\def\N{{\mathbb{N}}}

\def\k{{\rm K}}

\def\M{{\bf M}}
\def\F{{\bf F}}

\def\lm{{\rm lm}}
\def\lc{{\rm lc}}

\def\e{{\bf e}}

\def\u{{\bf u}}
\def\v{{\bf v}}

\def\lcm{{\rm lcm}}

\def\deg{{\rm deg}}
\def\max{{\rm max}}

\def\zero{{\bf 0}}

\def\d{{\tt deg}}
\def\jpset{{\tt JP}}

\def\done{{\tt Done}}
\def\goto{{\tt Goto}}
\def\step{{\tt Step}}
\def\ind{{\tt index}}
\def\rem{{\tt Rem}}
\def\jpd{{\tt JPdeg}}
\def\sym{{\tt SymbolicProcess} }
\def\eli{{\tt Elimination} }

\def\lla{{\longleftarrow}}

\newcommand{\spc}{\hspace*{15pt}}

\newcommand{\comment}[1]{}
\newcommand{\ignore}[1]{}

\begin{document}

\begin{frontmatter}



\title{An Improvement over the GVW Algorithm for Inhomogeneous Polynomial Systems\tnoteref{sponsor}
}

\tnotetext[sponsor]{The authors are supported by National Key Basic Research Program of China (No. 2013CB834203), National Natrue Science Foundation of China (No. 11301523), the Strategic Priority Research Program of the Chinese Academy of Sciences (No. XDA06010701), and IEE's Research Project on Cryptography (No. Y3Z0013102).}


\author[sklois]{Yao Sun}
\ead{sunyao@iie.ac.cn}


\author[sklois]{Dongdai Lin}
\ead{ddlin@iie.ac.cn}

\author[klmm]{Dingkang Wang}
\ead{dwang@mmrc.iss.ac.cn}

\address[sklois]{SKLOIS, Institute of Information Engineering, CAS, Beijing 100093,  China}

\address[klmm]{KLMM, Academy of Mathematics and Systems Science, CAS, Beijing 100190, China}

\begin{abstract}
The GVW algorithm is a signature-based algorithm for computing Gr\"obner bases. If the input system is not homogeneous, some J-pairs with higher signatures but lower degrees are rejected by GVW's Syzygy Criterion, instead, GVW have to compute some J-pairs with lower signatures but higher degrees. Consequently, degrees of polynomials appearing during the  computations may unnecessarily grow up higher and the computation become more expensive. In this paper, a variant of the GVW algorithm, called M-GVW, is proposed and mutant pairs are introduced to overcome inconveniences brought by inhomogeneous input polynomials. Some techniques from linear algebra are used to improve the efficiency. Both GVW and M-GVW have been implemented in C++ and tested by many examples from  boolean polynomial rings. The timings show M-GVW usually performs much better than the original GVW algorithm when mutant pairs are found. Besides,  M-GVW is also  compared with intrinsic Gr\"obner bases functions on Maple, Singular and Magma. Due to the efficient routines from the M4RI library, the experimental results show that M-GVW is very efficient.
\end{abstract}


\end{frontmatter}



\section{Introduction}
Gr\"obner bases, proposed by Buchberger in 1965 \citep{Buchberger65}, have been proven to be very useful in many aspects of algebra. In the past forty years, many efficient algorithms have been proposed to compute Gr\"obner bases. One important improvement is that Lazard pointed out the strong relation between Gr\"obner bases and linear algebra \citep{Lazard83}. This idea has been implemented  in F4 by Faug\`ere\citep{Fau99}, and also as XL type algorithms by Courtois et al. \citep{Courtois00} and Ding et al. \citep{Ding08}.

Faug\`ere introduced the concept of signatures for polynomials and presented the famous F5 algorithm \citep{Fau02}.
Since then, signature-based algorithms have been widely investigated, and several variants of F5  have been presented, including F5C \citep{Eder09}, extended F5 \citep{Ars09}, F5 with revised criterion (the AP algorithm) \citep{Arri11}, and RB \citep{Eder13b}.  Gao et al. proposed another signature based algorithm G2V \citep{Gao09} in a different way from F5, and GVW\citep{Gao10} (which is unpublished) is an extended version of G2V. The authors also studied generalized criteria and signature-based algorithms in solvable polynomial algebra in \citep{SunWang11, SunWang12}.

In GVW, criteria reject J-pairs with higher signatures, and process J-pairs with lower signatures instead. However, when input systems are inhomogeneous, J-pairs with higher signatures do not always have higher degrees, where by saying degrees of polynomials, we mean the total degrees of polynomials.  This is not good for Gr\"obner basis computations, and particularly not good for an implementation of GVW using linear algebra, because as suggested by Faug\`ere in \citep{Fau99, Fau02}, a good strategy of dealing with critical pairs (equivalent to J-pairs in GVW) in a batch is to select all critical pairs with the minimal degree. The reason is that critical pairs with higher degrees usually lead to larger matrices, which will cost much more time for eliminations. Some influences of inhomogeneous input systems were also discussed by Eder \citep{Eder13}.

We find that with GVW's Syzygy Criterion is possible to reject J-pairs with higher signatures but lower degrees such that GVW has to compute J-pairs with lower signatures but higher degrees.  After analysis,  we find  that such phenomenons are caused by some mutant  pairs, which will be defined in Section \ref{sec_vgvw}, and then  we propose a variant of the GVW algorithm (called M-GVW).
In M-GVW, when mutant pairs are found during the computations, we will append them to the initial input system  and assign new signatures to such mutant pairs. In this way, J-pairs generated by mutant pairs will not be all rejected by GVW's Syzygy Criterion, and hence, the maximal degree of polynomials appearing in the computations will not become too high. Particularly, for homogeneous polynomial systems, no mutant pairs will be generated and M-GVW is exactly the GVW algorithm.

 For implementations of signature-based algorithms, Roune and Stillman  efficiently implemented  GVW and  AP without using linear algebra \citep{Roune12}. Faug\`ere mentioned a matrix F5 in \citep{Fau09}.  A matrix F5 was described in more detail in an unpublished paper by Albrecht and Perry \citep{Albrecht10}.
We have implemented both the original GVW and M-GVW with linear algebra over boolean polynomial rings. For eliminations of matrices, we hope to take the advantage of  fast arithmetics of dense matrices over GF(2) provided by the library M4RI \citep{Albrecht13}. However, reductions in signature-based algorithms must be done in one direction, i.e. rows with higher signatures can only be eliminated by rows with lower signatures. So functions from M4RI cannot be used directly. We propose a method to do such one-direction eliminations for dense matrices by modifying  functions from M4RI in our implementations.

The timings show M-GVW usually performs much better than the original GVW algorithm when mutant pairs are found. Besides,  M-GVW is also  compared with intrinsic Gr\"obner bases functions on Maple, Singular and Magma. The experimental results show that  M-GVW is very efficient.

This paper is organized as follows. In Section \ref{sec_vgvw}, we revisit the GVW algorithm and present M-GVW on a theoretical level. In Section \ref{sec_implementation}, we discuss details on implementing M-GVW over boolean polynomial rings. Some experimental results are shown in Section \ref{sec_timings}. Conclusion remarks follow in Section \ref{sec_conclusion}.

\section{A variant of the GVW algorithm} \label{sec_vgvw}

In this section, we present a variant of GVW (M-GVW) in theoretical level.  We give this new algorithm over general polynomial rings, i.e.  with no special assumptions on ground fields and monomial orderings.



\subsection{The GVW algorithm revisited} \label{subsec_gvw}

Most of notations and definitions are inherited from Gao et al.'s original paper. For more details, please see \citep{Gao10}.

Let $R=\k[x_1, \ldots, x_n]$ be a polynomial  ring over a field $\k$ with $n$ variables, and $\{f_1, \cdots, f_m\}$ is a finite subset of $R$. We want to compute a Gr\"obner basis for the ideal
$$I = \langle f_1, \cdots, f_m\rangle = \{p_1f_1+\cdots+p_mf_m \mid p_1,\cdots,p_m \in R\}$$ with respect to some monomial ordering on $R$.

Let $\F=(f_1, \cdots, f_m) \in R^m$, and consider the following $R$-module of $R^m\times R$: $$\M = \{(\u, f) \in R^m \times R \mid \u \cdot \F=f\}.$$ Let $\e_i$ be the $i$-th unit vector of $R^m$, i.e. $(\e_i)_j=\delta_{ij}$ where $\delta_{ij}$ is the Kronecker delta. Then the $R$-module $\M$ is generated by $\{(\e_1, f_1), \cdots, (\e_m, f_m)\}.$ 

A monomial in $R$ has the form $x^\alpha = \Pi_{i=1}^n x_i^{a_i}$, where $\alpha = (a_1, \ldots, a_n) \in \N^n$ and $\N$ is the set of all non-negative integers. A monomial in $R^m$ is of the form $x^\alpha\e_i$, where $1 \le i \le m$ and $\alpha \in \N^n$. For monomials in $R^n$, we say $x^\alpha \e_i$ divides $x^\beta \e_j$ (or $x^\alpha \e_i \mid x^\beta \e_j$ for short), if $i = j$ and $x^\alpha$ divides $x^\beta$, and the quotient is defined as $(x^\beta\e_i)/(x^\alpha\e_i) = x^{\beta - \alpha}\in R$.

Fix any monomial ordering $\prec_p$ on $R$ and any monomial ordering $\prec_s$ on $R^m$ (subscripts $p$ and $s$ stand for {\em polynomial} and {\em signature} respectively). Please note that $\prec_s$ may or may not be related to $\prec_p$ in theory, although we always assume $\prec_s$ is {\bf compatible} with $\prec_p$ practically, i.e. $x^\alpha \prec_p x^\beta$ if and only if $x^\alpha\e_i \prec_s x^\beta \e_i$ for $1 \le i \le m$. To make descriptions simpler, we use the following notations for leading monomials: $$\lm(f)=\lm_{\prec_p}(f) \mbox{ and } \lm(\u)=\lm_{\prec_s}(\u),$$ for any $f\in R$ and any $\u\in R^m$. Leading monomials of $f\in R$ and $\u\in R^m$ are monomials without coefficients in $R$ and $R^m$ respectively. We define $\lm(f)=0$ if $f=0$, and $0 \prec_p x^\alpha$ for any non-zero monomial $x^\alpha$ in $R$; similarly for monomials in $R^m$. In the rest of this paper, we use $\prec$ to represent $\prec_p$ and $\prec_s$ for short, if no confusion occurs.

For a pair $(\u, f)\in \M$,  $\lm(\u)$ is called the {\bf signature} of $(\u, f)$. This definition is the same as that used in GVW, but different from those used in \citep{Fau02, Arri11}. The difference is discussed in \citep{Gao10}.

Let $(\u, f)\in \M$ and $B\subset \M$, we say $(\u, f)$ is {\bf top-reducible} by $B$, if there exists $(\v, g)\in B$ with $g\not=0$, such that $\lm(g)$ divides $\lm(f)$ and $\lm(\u) \succeq \lm(t\v)$ where $t=\lm(f)/\lm(g)$. The corresponding {\bf top-reduction} is then $$(\u, f) - ct(\v, g) = (\u - ct\v, f - ctg),$$ where $c = \lc(f)/\lc(g)$. Particularly, this top-reduction is called {\bf regular}, if $\lm(\u) \succ \lm(t\v)$; and {\bf super} if $\lm(\u) = \lm(t\v)$.\footnote{Regular top-reduction defined here is slightly different from its original version in \citep{Gao10}, but this will not affect proofs of related propositions and theorems.} Clearly, $(\u - c t \v, f-c t g)$ is also an element in $\M$.

A subset $G$ of $\M$ is called a {\bf strong Gr\"obner basis} for $\M$ if every nonzero pair $($pairs $\not=(\zero, 0))$ in $\M$ is top-reducible by  $G$. By Proposition 2.2 of \citep{Gao10}, let $G = \{(\v_i, g_i) \mid 1\le i \le s\}$ be a strong Gr\"obner basis for $\M$. Then $\{g_i: 1\le i \le s\}$ is a Gr\"obner basis for $I = \langle f_1, \ldots, f_m\rangle$.

Next, we define {\em joint pairs/J-pairs}. Suppose $(\u, f), (\v, g)\in \M$ are two pairs with $f$ and $g$ both nonzero. Let $t=\lcm(\lm(f), \lm(g))$, $t_f=t/\lm(f)$ and $t_g=t/\lm(g)$. Then the {\bf J-pair} of $(\u, f)$ and $(\v, g)$ is defined as: $t_f(\u, f)$ (or $t_g(\v, g)$),  if $\lm(t_f\u) \succ \lm(t_g\v)$ (or $\lm(t_f\u) \prec \lm(t_g\v)$). For the case $\lm(t_f\u) = \lm(t_g\v)$, the J-pair is not defined. Note that the J-pair of $(\u, f), (\v, g)\in \M$ is also a pair in $\M$. Assume $t_f(\u, f)$ is the J-pair of $(\u, f)$ and $(\v, g)$, the {\bf degree} of $t_f(\u, f)$ is defined as $\deg(t_f f)$, i.e. the degree of the polynomial part. For convenience, we call a J-pair is of $G \subset \M$, if this J-pair is the J-pair of two pairs in $G$.

For a pair $(\u, f) \in \M$ and a set $G \subset \M$, we say $(\u, f)$ is {\bf covered} by $G$, if there is a pair $(\v, g)\in G$, such that $\lm(\v)$ divides $\lm(\u)$ and $t\lm(g) \prec \lm(f)$ $($strictly smaller$)$ where $t = \lm(\u)/\lm(\v)$.


Two criteria are used in the GVW algorithm.\smallskip

\noindent {\bf [Syzygy Criterion]}
For a J-pair $t_f(\u, f)$ of a set $G\in \M$, if there exist $(\v, 0)\in G$ such that $\lm(\v)$ divides $t_f\lm(\u)$, then this J-pair can be discarded.\smallskip

\noindent {\bf [Second Criterion]}
For a J-pair of a set $G\in \M$, if this J-pair is covered by $G$, then this J-pair can be discarded.\smallskip

In this paper, we call the second criterion {\bf Rewriting Criterion}. Arri and Perry proposed a quite similar criterion to Rewriting Criterion in \citep{Arri11}.  Comments on Arri-Perry's criterion and Rewriting Criterion can be found in \citep{Gao10, Roune12}.

The following GVW algorithm is slightly modified from its original version. We delete the output of a Gr\"obner basis for the syzygy module of input polynomials, because we only care about the Gr\"obner basis of input polynomials in current paper. We emphasize that for a pair $(\u, f) \in \M$, only $(\lm(\u), f)$ is stored in the latest version of the GVW algorithm. Related conceptions, such as {\em top-reduction, J-pairs} and {\em cover}, are defined similarly. Please see \citep{Gao10} for more details.

\begin{algorithm}[!ht]
\DontPrintSemicolon
\SetAlgoSkip{}
\LinesNumbered

\SetKwInOut{Input}{Input}
\SetKwInOut{Output}{Output}
\SetKwFor{For}{for}{do}{end\ for}
\SetKwIF{If}{ElseIf}{Else}{if}{then}{else\ if}{else}{end\ if}

\Input{$f_1, \ldots, f_m \in R = K[x_1, \ldots, x_n]$, monomial orderings for $R$ and $R^m$.}

\Output{A Gr\"obner basis of $I = \langle f_1,\ldots, f_m\rangle$.}

\BlankLine

\Begin{

$H \lla \{\lm(f_j\e_i - f_i\e_j) \mid 1 \le i, j \le m\}$

$G \lla \{(\lm(\e_i), f_i) \mid 1\le i \le m\}$

$\jpset \lla \{$all J-pairs of $G\}$

\While{$\jpset \not= \emptyset$}
{
Let $t(x^\alpha\e_i, f) \in \jpset$ \hspace{2cm} $(\bigstar)$

$\jpset \lla \jpset \setminus \{t(x^\alpha\e_i, f)\}$

\If{$tx^\alpha\e_i$ is divisible by monomials in $H$}{\goto \step 5 (Syzygy Criterion)}

\If{$t(x^\alpha\e_i, f)$ is covered by $G$}{\goto \step 5 (Rewriting Criterion)}

$(x^\gamma\e_i, h) \lla${\em Regular} top-reduce $(tx^\alpha\e_i, tf)$ by $G$.

\If{$h = 0$}
{$H \lla H \cup \{x^\gamma\e_i\}$

\goto \step 5
}

\For{$(x^\beta\e_j, g) \in G$ s.t. $\lm(g)x^\gamma\e_i\not= \lm(h)x^\beta\e_j$}
{$H \lla H \cup \{\max(\lm(g)x^\gamma\e_i, \lm(h)x^\beta\e_j) \}$

$\jpset \lla \jpset \cup \{$J-pair of $(x^\gamma\e_i, h)$ and $(x^\beta\e_j, g)\}$
}
$G \lla G \cup \{(x^\gamma\e_i, h)\}$
}

{\bf return} $\{g \mid (x^\beta\e_j, g)\in G\}$
}
\caption{The GVW algorithm}
\end{algorithm}


There are some remarks on the GVW algorithm.

\begin{enumerate}

\item At Step 6 (marked with black star), a J-pair can be selected from JP in {\bf any} order. In Section \ref{sec_implementation}, we prefer to choosing J-pairs with minimal degrees first.

\item Proposition 2.2 in \citep{Gao10} ensures correctness of GVW when J-pairs are computed in {\bf any} order.

\item The finite termination of GVW is proved by Theorem 3.1 in \citep{Gao10} when monomial orderings of $R$ and $R^m$ are compatible. Particularly, GVW also terminates in finite steps when J-pairs are computed in {\bf any} order. This proof is first given by Theorem 3.5 in \citep{SunWang12}.

\item The GVW algorithm in \citep{Gao10} retains only one J-pair (the one with the minimal polynomial part) when there are several J-pairs having the same signature. This process can be implied by the ``cover check'' at step 10.

\end{enumerate}

\subsection{Motivation and main ideas} \label{subsec_motivation}

The motivation of varying GVW arises when we are implementing GVW with linear algebra in boolean polynomial rings. To control the size of appearing matrices as small as possible, we deal with J-pairs with the {\em minimal degree} first. That is, at Step 6 of GVW, we find the minimal degree of all J-pairs in JP first, and then choose the J-pair with the smallest signature among J-pairs with the minimal degree.

However, we are quite surprised to find that when computing a Gr\"obner basis for HFE$\_$25$\_$96 from \citep{Steel04}, degrees of matrices always grow up to 5, this makes our implementation much less efficient, because the sizes of matrices with degree 5 are much larger than those with degree 4, and moreover, it has been shown in \citep{Fau03} that Gr\"obner basis of this example can definitely be obtained with degrees of matrices smaller than 5.
 Here the degree  of a matrix is  the maximal degree of the polynomials to construct this matrix.

We find the above phenomenon does not depend on the computing orders of J-pairs.
In order to illustrate this phenomenon clearly, we finally get  the following example
after testing many examples.
\begin{example}
Let $\{f_1, f_2, \ldots, f_{11}\} \subset R = \gf_2[x_1, x_2, \ldots, x_9]$, where $\gf_2$ is the Galois Field $GF(2)$, and $$f_1 = x_1x_2x_5x_6+x_2x_3x_7x_9+x_7, f_2 = x_1x_2x_6x_8+x_3x_4x_7,$$ $$f_{i+2} = x_i^2+x_i, \mbox{ for } 1 \le i \le 9.$$ Monomial ordering $\prec_p$ in $R$ is the Graded Reverse Lexicographic ordering, and $\prec_s$ in $R^{11}$ is a position over term extension of $\prec_p$: $$x^\alpha\e_i \prec_s x^\beta\e_j \mbox{ iff } i> j,  \mbox{ or }  i = j \mbox{ and } x^\alpha \prec_p x^\beta.$$ Thus, $\e_1 \succ \e_2 \succ \cdots \succ \e_{11}$.
\end{example}

For this example, GVW needs to deal with J-pairs having degree bigger than 5, while the maximal degree of matrices in the F4 algorithm with criteria from \citep{Buchberger79} is only 5. This implies some ``useful'' J-pairs with degrees not bigger than 5 have been rejected by GVW's criteria.

Now,  we will discuss this  example in details.
We compute a Gr\"obner basis for $\langle f_1, \ldots, f_{11}\rangle$ by GVW with the following strategy for selecting J-pairs at Step 6:
\begin{enumerate}

\item $deg \lla$the minimal degree of J-pairs in JP.

\item $(x^\alpha\e_i, f) \lla$ J-pair with the smallest signature in the set $\{(x^\beta\e_j, g) \in$ JP $\mid \deg(g) = deg\}$.

\end{enumerate}
Since the  $f_i$'s are all inhomogeneous,   the above strategy   leads to that J-pairs  are not handled in an increasing order on signatures\footnote{Even dealing with J-pairs in an increasing order on signatures, GVW still has to reduce J-pairs with degrees bigger than 5 before a strong Gr\"obner basis is obtained.}.


Due to page limitation, we only present some results here. Initially, we have $G = \{(\e_1, f_1), (\e_2, f_2), \ldots, (\e_{11}, f_{11})\}$.
Before dealing with J-pairs with degree 6, the following polynomials are generated one by one: \smallskip \\
$(x_8\e_2, f_{12} = x_3x_4x_7x_8+x_3x_4x_7)$,\\
$(x_6\e_2, f_{13} = x_3x_4x_6x_7+x_1x_2x_6x_8),$ \\
$(x_2\e_2, f_{14} = x_2x_3x_4x_7+x_1x_2x_6x_8)$,\\
$(x_1\e_2, f_{15} = x_1x_3x_4x_7+x_1x_2x_6x_8)$,\\
$(x_8\e_1, f_{16} = x_2x_3x_7x_8x_9+x_3x_4x_5x_7+x_7x_8)$,\\
$(x_6\e_1, f_{17} = x_2x_3x_6x_7x_9+x_1x_2x_5x_6+x_6x_7)$,\\
$(x_5\e_1, f_{18} = x_2x_3x_5x_7x_9+x_1x_2x_5x_6+x_5x_7)$,\\
$(x_2\e_1, f_{19} = x_2x_7+x_7)$,\\
$(x_1\e_1, f_{20} = x_1x_2x_3x_7x_9+x_1x_2x_5x_6+x_1x_7)$,\\
$(x_2x_3x_8x_9\e_1, f_{21} = x_3x_4x_5x_7+x_3x_7x_8x_9+x_7x_8)$,\\
$(x_2x_3x_6x_9\e_1, f_{22} = x_3x_6x_7x_9+x_3x_7x_9+x_6x_7+x_7)$,\\
$(x_2x_3x_5x_9\e_1, f_{23} = x_3x_5x_7x_9+x_3x_7x_9+x_5x_7+x_7)$,\\
$(x_1x_2x_3x_9\e_1, f_{24} = x_1x_3x_7x_9+x_2x_3x_7x_9+x_1x_7+x_7)$.\smallskip

During computations, many leading monomial of syzygies in $\M$ are generated, among them the one $x_2x_3x_4\e_1$ (obtained before $f_{20}$) is important, since it has been used to reject many other J-pairs.

So far, all J-pairs with degrees not bigger than 5 have been considered. It is easy to check $\{f_1, f_2, \ldots, f_{24}\}$ is {\em not} a Gr\"obner basis of $\langle f_1, f_2, \ldots, f_{11}\rangle$. However, for the same ideal, F4 algorithm with criteria from \citep{Buchberger79} can obtain a Gr\"obner basis without computing any critical pairs with degrees bigger than 5. Comparing GVW and F4 algorithm step by step, finally, we find the following J-pairs:
\vspace*{-7pt}
\begin{table}[h]
$x_4(x_2x_3x_5x_9\e_1, f_{23})$, \spc\spc $x_4(x_1x_2x_3x_9\e_1, f_{24})$,

$x_4(x_2x_3x_6x_9\e_1, f_{22})$, \spc\spc $x_1x_5x_6(x_2\e_1, f_{19})$,

$x_1x_6x_8(x_2\e_1, f_{19} )$, \spc\spc $x_3(x_2x_3x_8x_9\e_1, f_{21})$,

 $x_4(x_2x_3x_8x_9\e_1, f_{21})$, \spc\spc$x_3(x_1x_2x_3x_9\e_1, f_{24})$,

$x_9(x_1x_2x_3x_9\e_1, f_{24})$, \spc\spc$x_3(x_2x_3x_5x_9\e_1, f_{23})$,

$x_9(x_2x_3x_5x_9\e_1, f_{23})$, \spc\spc$x_3(x_2x_3x_6x_9\e_1, f_{22})$,

$x_9(x_2x_3x_6x_9\e_1, f_{22})$,
\end{table}
\vspace*{-7pt}

\noindent which are rejected by GVW's criteria but not rejected by Buchberger's criteria in F4.

Reducing these J-pairs,
we get 9 polynomials with degree 3 and 4 polynomials with 4. These polynomials are computed in F4, and prevent F4 to deal with critical pairs with degree bigger than 5.

Next, we analyze why GVW is possible to reject J-pairs with lower degrees and prefer to reducing higher degree J-pairs.

Take J-pair $x_3(x_1x_2x_3x_9\e_1, f_{24})$ for example. Reducing this J-pair, we get $$(x_1x_2x_3^2x_9\e_1, x_1x_3x_7+x_1x_7+x_3x_7+x_7).$$ But this J-pair is rejected by $((x_3^2 + x_3)\e_1 - f_1\e_5, 0) \in \M$ in GVW, which is the principal syzygy  of $f_1$ and $f_5$. While running GVW forward, we find the polynomial $x_1x_3x_7+x_1x_7+x_3x_7+x_7$ is obtained from the J-pair $x_3(x_1\e_1, f_{20})$, which is a J-pair of degree 6. Combined with our experiences of proving F5 in \citep{SunWang13}, we have the following observation.

\begin{remark}
GVW's criteria alway reject J-pairs with higher signatures, and proceed some J-pairs with smaller signatures instead.

When input systems are inhomogeneous, J-pairs with bigger signatures may have lower {\bf degrees} than J-pairs with smaller signatures.
\end{remark}

Consider the degree-5 J-pair $x_3(x_1x_2x_3x_9\e_1, f_{24})$ again. The syzygy $((x_3^2 + x_3)\e_1 - f_1\e_5, 0) \in \M$, which rejects this J-pair, corresponds to the
equation $$f_5 f_1 - f_1 f_5 = 0,$$ in which monomials with degree 6 appear.  Thus, we believe that  if we use the syzygy $((x_3^2 + x_3)\e_1 - f_1\e_5, 0)$ to reject this
degree-5 J-pair, it is possible to deal with some J-pairs involving polynomials of degree 6 instead later. On seeing this, our basic idea is to prevent GVW rejecting J-pairs
like $x_3(x_1x_2x_3x_9\e_1, f_{24})$ by using syzygy like $((x_3^2 + x_3)\e_1 - f_1\e_5, 0)$.

Analyzing all J-pairs we have listed earlier, we find they have two common properties:
\begin{enumerate}

\item All these J-pairs are rejected by Syzygy Criterion.

\item For any J-pair $t(x^\alpha\e_1, f_{j})$ that is listed,  we find that $$\deg(x^\alpha) + \deg(f_1) > \deg(f_j).$$
\end{enumerate}
The second property   makes the degree of $x_3(x_1x_2x_3x_9\e_1, f_{24})$   lower than the degree of $f_5f_1$. In order to prevent this J-pair to be rejected,
there are two possible method: (1) treat the syzygy $(f_5\e_1 - f_1\e_5, 0)$ specially; and (2) treat the pair $(x_1x_2x_3x_9\e_1, f_{24})$ specially.

For the first method, we can store the degree of $f_5f_1$ together with the syzygy $(f_5\e_1 - f_1\e_5, 0)$ and prevent it to reject J-pairs that have lower degrees than $\deg(f_5f_1)$. But we do not use the first method in our implementation, because the second method seems to be simpler.


We find the pairs satisfying  the property 2 are similar to mutant polynomials defined in \citep{Ding08}, so we give the following definition.
\begin{define}
Let $\M$ be an $R$-module generated by $\{(\e_1, f_1)$, $\ldots, (\e_m, f_m)\}$. A pair $(\u, f)\in \M$ with $\lm(\u) = x^\alpha\e_i$ and $f\not = 0$, is called {\bf mutant}, if $\deg(x^\alpha) + \deg(f_i) > \deg(f)$.
\end{define}
Due to the existence of syzygy pairs, there are lots of mutant pairs in $\M$. But mutant pairs appearing in GVW are not so many, since the Syzygy Criterion is used.

\subsection{The M-GVW algorithm} \label{subsec_vgvw}

 The basic idea of M-GVW is to   append mutant pairs  to the initial input system and assign new signatures to such pairs so that  the J-pairs generated by mutant pairs will not be all rejected by GVW's Syzygy Criterion, and hence, the maximal degree of polynomials appearing in the computations will not become too high.

Specifically, let $\M$ be generated by $\{(\e_1, f_1)$, $\ldots, (\e_m, f_m)\}$, and $(\u, f)$ with $\lm(\u) = x^\alpha\e_k$ be the first mutant pair that we meet during computations where $\e_i \in R^m$. Then we add a pair $(\e_{m+1}, f)$ as a new generator and the module is expanded to the module generated by $\{(\e_1, f_1)$, $\ldots, (\e_m, f_m)$, $(\e_{m+1}, f_{m+1} = f)\}$.  Please note that dimensions of $\e_1, \ldots, \e_m$ are enlarged to $m+1$ by appending $0$'s to last entry, and now $\e_i \in R^{m+1}$.  We emphasize that, after appending $(\e_{m+1}, f_{m+1})$, we always require
\begin{equation}\label{equ_require}
x^\alpha\e_k \succ_s \e_{m+1}.
\end{equation}
That is, signature of a new appended generator $(\e_{m+1}, f_{m+1})$ should be smaller than the signature of the mutant pair $(\u, f)$, such that $(\e_{m+1}, f_{m+1})$ will not be reduced to 0. Next, when the second mutant pair is obtained, we append it as the $(m+2)$th generator, and so on.  This appending method was mentioned in \citep{SunWang09} by authors.

In order to ensure termination of this variant algorithm, when we meet a mutant polynomial $(\u, f)$, we usually do not append $f$ as the $k$-th generator directly. Instead, we compute the remainder of $f$ w.r.t. the previous  $k-1$ existing generators $\{f_1, \ldots, f_{k-1}\}$ by polynomial division first (defined in \citep{Cox05}) without consideration of signatures, and denote the remainder as $f'$. If $f'\not=0$, then $\lm(f')$ is not divisible by any $\lm(f_i)$ where $1 \le i < k$, and then we add $(\e_k, f')$ as the $k$-th generator.
Please note that $f'$ is in the ideal generated by $\{f_1,\ldots, f_{k-1} \}$.

Next we give M-GVW below. Function $Rem(f, F)$ computes a remainder of $f$ w.r.t. $F$ by polynomial division.

\begin{algorithm}[!ht]
\DontPrintSemicolon
\SetAlgoSkip{}
\LinesNumbered

\SetKwInOut{Input}{Input}
\SetKwInOut{Output}{Output}
\SetKwFor{For}{for}{do}{end\ for}
\SetKwIF{If}{ElseIf}{Else}{if}{then}{else\ if}{else}{end\ if}

\Input{$f_1, \ldots, f_m \in R = K[x_1, \ldots, x_n]$, monomial orderings for $R$ and $R^m$.}

\Output{A Gr\"obner basis of $I = \langle f_1,\ldots, f_m\rangle$.}

\BlankLine

\Begin{

$H \lla \{\lm(f_j\e_i - f_i\e_j) \mid 1 \le i, j \le m\}$

$G \lla \{(\lm(\e_i), f_i) \mid 1\le i \le m\}$

$\ind \lla m$

$\jpset \lla \{$all J-pairs of $G\}$

\While{$\jpset \not= \emptyset$}
{
Let $t(x^\alpha\e_i, f) \in \jpset$

$\jpset \lla \jpset \setminus \{t(x^\alpha\e_i, f)\}$

\If{$tx^\alpha\e_i$ is divisible by monomials in $H$}{\goto \step 5 (Syzygy Criterion)}

\If{$t(x^\alpha\e_i, f)$ is covered by $G$}{\goto \step 5 (Rewriting Criterion)}

$(tx^\alpha\e_i, h) \lla${\em Regular} top-reduce $t(x^\alpha\e_i, f)$ by $G$.

\If{$h = 0$}
{$H \lla H \cup \{tx^\alpha\e_i\}$

\goto \step 5
}

\If{$\deg(tx^\alpha) + \deg(f_i) > \deg(h)$ and $\rem(h, \{f_1, \ldots, f_{\ind}\}) \not= 0$}
{

$f_{\ind + 1} \lla \rem(h, \{f_1, \ldots, f_{\ind}\})$

$\ind \lla \ind + 1$

Denote $(\e_{\ind}, f_{\ind})$ as $(x^\gamma\e_k,  p )$
}
\Else{
Denote $(tx^\alpha\e_i, h)$ as $(x^\gamma\e_k, p)$
}
\For{$(x^\beta\e_j, g) \in G$, $\lm(g)x^\gamma\e_k\not= \lm(p)x^\beta\e_j$}
{$H \lla H \cup \{\max(\lm(g)x^\gamma\e_k, \lm(p)x^\beta\e_j) \}$

$\jpset \lla \jpset \cup \{$J-pair of $(x^\gamma\e_k, p)$ and $(x^\beta\e_j, g)\}$
}
$G \lla G \cup \{(x^\gamma\e_k, p)\}$
}

{\bf return} $\{g \mid (x^\beta\e_j, g)\in G\}$
}
\caption{The M-GVW algorithm}
\end{algorithm}

\begin{theorem}\label{thm_mgvwterminate}
The M-GVW algorithm terminates in finite steps, if monomial orderings in $R^m$ and $R$ are compatible.
\end{theorem}

\begin{proof}
At step 17-20 of M-GVW, a new generator is appended when $(tx^\alpha\e_i, h)$ is mutant and the remainder of $h$, say $h'$, w.r.t. $\{f_1, \ldots, f_{index}\}$ is not $0$. If $h'\not=0$, the ideal $\langle \lm(f_1), \ldots, \lm(f_{index}) \rangle$ is strictly smaller than $\langle \lm(f_1)$, $\ldots, \lm(f_{index}), \lm(h') \rangle$. Ascending chain  condition of ideals \citep{Cox05} implies M-GVW can only append finite many new generators. That is, after appending some generator $(\e_l, f_l)$, no more generators will be appended. In this case, M-GVW turns to be GVW, and the termination is ensured by Theorem 3.1 of \citep{Gao10}.
\end{proof}

\begin{theorem}\label{thm_mgvwcorrect}
The M-GVW algorithm is correct.
\end{theorem}

\begin{proof}
Clearly, we have $f_{index+1} \in \langle f_1, \ldots, f_{index}\rangle$ for $index \ge m$. Assume $(\e_l, f_l)$ is the last generator appended in M-GVW. In this case, M-GVW turns to be GVW, and
M-GVW computs a Gr\"obner basis for $\langle f_1, \ldots, f_l\rangle = \langle f_1, \ldots, f_m\rangle$ by Theorem 2.2 of \citep{Gao10}.
\end{proof}

There are some remarks on M-GVW.

\begin{enumerate}

\item Since we always make requirement like (\ref{equ_require}), we prefer $\prec_s$ to be a {\em position over term extension} of $\prec_p$ in M-GVW with $\e_1 \succ \e_2 \succ \cdots \succ \e_m \succ \cdots$.

\item In practical implementation, we usually do not append all mutant pairs as new generators. Because appending generators with high degrees often make the implementation less efficient, and too many generators will also weaken the power of Syzygy Criterion.
So we usually add a constraint ``$\deg(h) <$ {\rm  Deg-Limit} '' at Step 17, where {\rm  Deg-Limit}   is a given constant.



\item Mutant pairs cannot be found in M-GVW when input systems are {\em homogeneous}. In this case, M-GVW is just the GVW algorithm.

\end{enumerate}

\section{An implementation with linear algebra over boolean polynomial rings} \label{sec_implementation}

In this section, we give an implementation of M-GVW based on the dense matrix library M4RI, and show some details in our implementation.

The polynomial ring is specialized as $R = \gf_2[x_1, x_2, \ldots, x_n]$ with $n$ variables over the Galois Field $GF(2)$. Polynomials $E = \{x_1^2+x_1, \ldots, x_n^2+x_n\}$ are called field polynomials. Let $F = \{f_1, \ldots, f_m\}$ be a subset of $R$. Then computing a Gr\"obner basis for the ideal generated by $F$ over the boolean polynomial ring $R/\langle E \rangle$, is equivalent to computing a Gr\"obner basis for the ideal generated by $F \cup E$ over $R$. In our implementation, we aim to compute Gr\"obner bases for $F\cup E$ over $R$, so all the operations are done in $R$. In fact, since field polynomials have quite special forms, we do not need to store them in practical implementations, moreover, normal forms of polynomials in $R$ w.r.t. $E$ are also done automatically.

We specialize the monomial ordering $\prec_p$ on $R$ to be the {\em Graded Reverse Lexicographic ordering}. And monomial ordering $\prec_s$ on modules is a {\em position over term extension} of $\prec_p$, such that $\e_1 \succ_s \e_2 \succ_s \cdots$. Note that $\e_j$'s corresponding to field polynomials are always smaller than other non-field polynomials, even if new generators are appended, such that field polynomials can always be used for reductions.

In Subsection \ref{subsec_matrixstyle}, we write M-GVW in a matrix style. In Subsection \ref{subsec_elimination}, we show how to do reductions efficiently based on matrices.  

\subsection{M-GVW in matrix style} \label{subsec_matrixstyle}

The matrix version of M-GVW is quite similar to the F4 algorithm. The main function is given below.

\begin{algorithm}[!ht]
\DontPrintSemicolon
\SetAlgoSkip{}
\LinesNumbered

\SetKwInOut{Input}{Input}
\SetKwInOut{Output}{Output}
\SetKwFor{For}{for}{do}{end\ for}
\SetKwIF{If}{ElseIf}{Else}{if}{then}{else\ if}{else}{end\ if}

\Input{$f_1, \ldots, f_m \in R = K[x_1, \ldots, x_n]$, monomial orderings for $R$ and $R^m$.}

\Output{A Gr\"obner basis of $I = \langle f_1,\ldots, f_m\rangle$.}

\BlankLine

\Begin{

$H \lla \{\lm(f_j\e_i - f_i\e_j) \mid 1 \le i, j \le m\}$

$G \lla \{(\lm(\e_i), f_i) \mid 1\le i \le m\}$

$\ind \lla m$

$\jpset \lla \{$all J-pairs of $G\}$

\While{$\jpset \not= \emptyset$}
{
$\d \lla$ the minimal degree of J-pairs in $\jpset$

$\jpd \lla$ all J-pairs with degree $\d$ in $\jpset$

$\jpset \lla \jpset \setminus \jpd$

$\jpd' \lla$ discard J-pairs that are rejected by Syzygy and Rewritting Criteiron from $\jpset$

$P\lla \sym(\jpd', G)$ {\em (element in $P$ has a form of $(x^\alpha\e_i, f))$}

$F \lla \eli(P)$ {\em (element in $F$ has a form of $(x^\alpha\e_i, h))$}

$F^+ \lla F \setminus \{$pairs are super top-reducible by $G\}$

\For{each $(x^\alpha\e_i, h)\in F^+$ s.t. $h = 0$}
{$H \lla H \cup \{x^\alpha\e_i\}$
}

\For{each $(x^\alpha\e_i, h)\in F^+$ s.t. $h \not= 0$}
{
\If{$\deg(tx^\alpha) + \deg(f_i) > \deg(h)$ and $\deg(h) < $ {\rm  Deg-Limit} and $\rem(h, \{f_1, \ldots, f_{\ind}\}) \not= 0$}
{

$f_{\ind + 1} \lla \rem(h, \{f_1, \ldots, f_{\ind}\})$

$\ind \lla \ind + 1$

Denote $(\e_{\ind}, f_{\ind})$ as $(x^\gamma\e_k,  p )$
}
\Else{
Denote $(tx^\alpha\e_i, h)$ as $(x^\gamma\e_k, p)$
}
\For{$(x^\beta\e_j, g) \in G$, $\lm(g)x^\gamma\e_k\not= \lm(p)x^\beta\e_j$}
{$H \lla H \cup \{\max(\lm(g)x^\gamma\e_k, \lm(h)x^\beta\e_j) \}$

$\jpset \lla \jpset \cup \{$J-pair of $(x^\gamma\e_k, p)$, $(x^\beta\e_j, g)\}$
}
$G \lla G \cup \{(x^\gamma\e_k, p)\}$
}
}
{\bf return} $\{g \mid (x^\beta\e_j, g)\in G\}$
}
\caption{M-GVW in matrix style}
\end{algorithm}

Function $SymbolicProcess(JPdeg', G)$ will do three things. First, compute J-pairs from $JPdeg'$. Second, for each monomial that is not a leading monomial, search polynomials from $G$ to reduce it. Third, sort all pairs according their signatures, and if there are several pairs having the same signature, retain only one of them.  Denote $M(P)$ be the set of all monomials in $h$ for any $(x^\gamma\e_k, h)\in P$.

\begin{algorithm}[!ht]
\DontPrintSemicolon
\SetAlgoSkip{}
\LinesNumbered

\SetKwInOut{Input}{Input}
\SetKwInOut{Output}{Output}
\SetKwFor{For}{for}{do}{end\ for}
\SetKwIF{If}{ElseIf}{Else}{if}{then}{else\ if}{else}{end\ if}

\Input{$\jpset$, a set of J-pairs, $G$, a set of pairs.}

\Output{$P$, element in $P$ has a form of $(x^\alpha\e_i, f)$.}

\BlankLine

\Begin{

$P \lla \emptyset$

\For{each $t(x^\alpha\e_i, f)$ in $\jpset$}
{$P \lla P \cup \{(tx^\alpha\e_i, tf)\}$
}

$\done \lla \{\lm(h) \mid (x^\gamma\e_k, h) \in P\}$

\While{$M(P) \not= \done$}
{
$m \lla$ an element of $M(P) \setminus \done$

$\done \lla \done \cup \{m\}$

\If{$\exists (x^\beta\e_j, g) \in G$ s.t. (1) $\lm(g) \mid m$, and (2) $t(x^\beta\e_j, g)$ are not rejected by Syzygy and Rewritting Criterion, where $t = m/\lm(g)$}
{
$P \lla P \cup \{t(x^\beta\e_j, g)\}$
}
}

Sort $P$ by an increasing order on signatures, and if there are several pairs having the same signature, retain only one of them.

{\bf return} $P$

}
\caption{SymbolicProcess}
\end{algorithm}

This function is a bit different from Albrecht-Perry's version \citep{Albrecht10}. First, any monomial $m$ can be selected from $M(P)\setminus Done$, while in \citep{Albrecht10} the maximal one is selected each time.
Second, for any selected monomial $m$, we do not need to know any signature information about it.

Function $Elimination(P)$ will also do three things. First, write pairs in $P$ as rows of a matrix. Second, compute the echelon form of this matrix. Third, read polynomials from rows of this matrix. In the first step, building matrices from boolean polynomials is different from building matrices from general polynomials, because the product of a monomial and a boolean polynomial should be reduced by field polynomials automatically. We report our method in \citep{Sun13} and omit details here. The second step is critical. We do not use naive Gaussian eliminations, and want to use efficient divide-and-conquer eliminating methods from M4RI to improve efficiencies. However,  in signature-based algorithms, since rows with higher signatures can only be eliminated by rows with lower signatures, functions from M4RI can not be used directly. So we use a special kind of row swaps to replace original row swaps in M4RI, which will be discussed in the next subsection.

For doing criteria check, similarly as discussed in \citep{Albrecht10}, we maintain two arrays of ``rules'' for Syzygy and Rewriting Criteria respectively. In ``rules'' of Rewritting Criterion, we sort pairs according to ``ratios'' of pairs, which is first introduced in \citep{Roune12}.

\subsection{Elimination} \label{subsec_elimination}

Unlike eliminations in F4, eliminations in signature-based algorithm can only be done from one side. That is, each row of matrix has a signature, and rows with higher signatures can only be reduced by rows with lower signatures. Naive Gaussian eliminations can control eliminating directions easily.  But to use efficient divide-and-conquer strategy as well as efficient implementation of matrices multiplications in the library M4RI \citep{Albrecht13}, we allow eliminations to swap rows in a special manner. This strategy is a bit similar to the ideas in \citep{Dumas13}.

Let $A$ be a matrix  with entries in $\gf_2$. Assume $A$ has the following form:
\begin{center}
$\begin{array}{cc}
\begin{array}{c}
S_1 \\
S_2 \\
S_3 \\
S_4 \\
S_5 \\
S_6 \\
\end{array}
&
\left(
\begin{array}{cccccc}
0 & 1 & * & * & * & *\\
0 & 0 & 0 & 1 & * & *\\
0 & 0 & 1 & * & * & *\\
0 & 1 & * & * & * & *\\
\color{blue}{1} & * & * & * & * & *\\
1 & * & * & * & * & *\\
\end{array}
\right),
\end{array}$
\end{center}
where ``$*$'' may be 1 or 0, $S_i$ is the signature of each row, and we assume $S_1 \prec_s S_2 \prec_s \cdots \prec_s S_6$.

To reduce $A$ to row-echelon form, we first find the pivot entry in the first column. {\em We must search  the pivot entry from top to bottom  (i.e. from lower signatures to higher signatures)}. Then we find the entry at row 5 and col 1 is a pivot. If we use general methods of elimination, we need to swap row 1 and row 5 directly, and clear entries at column 1 by the row with signature $S_5$. Next, when doing elimination in the second column, the row with signature $S_4$ is selected as pivot row, and needs to eliminate other rows. However, this will leads to errors in signature-based algorithms, because the row with signature $S_1$ has a smaller signature than $S_4$ and cannot be eliminated by the row with signature $S_4$.  So we cannot swap row 1 and row 5 directly.

So to make further eliminations correct, we swap row 1 and row 5 in a special manner. First, we pick up the row 5 with signature $S_5$. Second, we move rows 4, 3, 2, and 1 to rows 5, 4, 3, and 2 respectively. At last, we put the row with signature $S_5$ at row 1. After this swap, matrix A becomes the following form.
\begin{center}
$\begin{array}{cc}
\begin{array}{c}
S_5 \\
S_1 \\
S_2 \\
S_3 \\
S_4 \\
S_6 \\
\end{array}
&
\left(
\begin{array}{cccccc}
\color{blue}{1} & * & * & * & * & *\\
0 & 1 & * & * & * & *\\
0 & 0 & 0 & 1 & * & *\\
0 & 0 & 1 & * & * & *\\
0 & 1 & * & * & * & *\\
1 & * & * & * & * & *\\
\end{array}
\right).
\end{array}$
\end{center}

Next, we use the row with $S_5$ to clear all entries at column 1 below this row, and then column 1 is done.  For column 2, we find pivots from rows with $S_1, ..., S_4$ and $S_6$, and repeat the above processes. Elimination terminates when the matrix becomes an upper triangular form.

This swap makes eliminations correct in signature-based algorithm for the following reasons.
On one hand, since pivot rows (e.g. row of $S_5$) are finding from low signatures to high signatures, rows with smaller signature (e.g. rows of $S_1, \ldots, S_4$) cannot be reduced by pivot rows (e.g. row of $S_5$). On the other hand, after swaps, rows below pivot rows (e.g. rows of $S_1, \ldots, S_4$ and $S_6$) are still in an increasing order on signatures.

Using this special swap, the echelon form of $A$ is in an upper triangular form, such that divide-and-conquer methods of PLE decomposition \citep{Albrecht11} can be used, and hence, the eliminations can be speeded up significantly.

In our implementation, we modify many subroutines of $mzd\_ple()$ in M4RI library to use this swap. The new function with the special swap is called $gvw\_ple()$. We compare the efficiency of $mzd\_ple()$ and $gvw\_ple()$ in the next section. The results show both functions almost have the same efficiency.

\section{Timings} \label{sec_timings}

M-GVW over boolean polynomial rings has been implemented in C++ and the library M4RI (version 20130416) is used \citep{Albrecht13}. The codes will be available at  \url{http://www.mmrc.iss.ac.cn/~dwang/software.html} sooner.

In Table 1, we test the efficiency of the function $gvw\_ple()$, which is modified from $mzd\_ple()$ by using the new swap method. Examples with density $\approx 50\%$ are generated randomly by routines from M4RI. Since the densities of matrices in Gr\"obner basis computations are usually very small, we also generate some randomized matrices with density $\approx 3\%$. The first column in Table 1 is the size of matrices, and the timings in this table are given by seconds.

{
\begin{table}[H]\centering

\medskip
\begin{tabular}{c | c c | c c }
\hline
Tests & \multicolumn{2}{c}{density $\approx 50\%$} &  \multicolumn{2}{c}{density $\approx 3\%$}\\

& mzd\_ple() & gvw\_ple() & mzd\_ple() & gvw\_ple() \\ \hline

$10, 000 \times 10, 000$ & 0.378 & 0.382 & 0.345 & 0.354 \\

$10, 000 \times 30, 000$ & 1.342 & 1.301 & 1.268 & 1.262 \\

$30, 000 \times 10, 000$ & 1.432 & 1.443 & 1.403 & 1.418 \\

$30, 000 \times 30, 000$ & 7.661 & 7.655 & 7.604 & 7.577 \\

$30, 000 \times 60, 000$ & 18.684 & 18.671 & 18.651 & 18.634 \\

$60, 000 \times 30, 000$ & 19.396 & 19.296 & 19.282 & 19.298 \\

$60, 000 \times 60, 000$ & 58.373 & 58.636 & 54.509 & 54.263 \\

$60, 000 \times 100, 000$ & 123.321 & 123.298 & 119.479 & 122.523 \\

$100, 000 \times 60, 000$ & 119.991 & 118.388 & 108.565 & 108.501 \\

$100, 000 \times 100, 000$ & 266.817 & 267.191 & 237.401 & 237.560 \\

$150, 000 \times 150, 000$ & 817.682 & 817.750 & 700.032 & 700.781 \\

\hline
\end{tabular}\\
\caption{\small mzd\_ple() vs gvw\_ple()}
\end{table}
}

From the above table, we can see the function $mzd\_ple()$ and $gvw\_ple()$ almost have the same efficiency. So the new swapping method presented in the last section does not slow down the efficiency of elimination.

In Table 2, we compare our implementations of the algorithm GVW and M-GVW. The size of maximal matrices generated during the computation and the timings are given in details, and the numbers of mutants pairs are also listed. In both  implementations, orderings of signatures are both position over term extensions of the Graded Reverse Lexicographic ordering. In M-GVW,
the parameter  {\rm  Deg-Limit} is set to 4.
In the column of Exam.,  $n\times n$  means that the input polynomial system has   $n$ polynomials with $n$ variables.
These  square polynomial systems were generated by   Courtois in  \citep{Courtois13}.  The left three  HFE systems are from \citep{Steel04}.
The Computer we used is MacBook Pro with 2.6 GHz Intel Core i7, 16 GB memory.

{
\begin{table}[H]\centering

\medskip
\begin{tabular}{c | c c | c c c}
\hline
Exam. & \multicolumn{2}{c}{GVW} &  \multicolumn{3}{c}{M-GVW}\\

& max mat.(deg) & time(sec) & max mat.(deg) & time(sec) & mutant pairs \\ \hline

$16 \times 16$ & $9378 \times 6861(5)$ & 0.560  & $9378 \times 6861(5)$ & 0.543 & 0  \\ 

$17 \times 17$ & $12012 \times 9354(5)$ & 0.893  & $12012 \times 9354 (5)$ & 0.895 & 0 \\ 

$18 \times 18$ & $15240 \times 12569(5)$  & 1.556 & $15240\times 12569(5)$ &  1.588 & 0 \\ 

$19 \times 19$ & $19043 \times 16613 (5)$  & 2.742 & $19043 \times 16613 (5)$ &  2.728 & 0 \\ 

$20 \times 20$ & $23478 \times 21640(5)$ & 4.676 & $23478 \times 21640(5)$  & 4.664 & 0  \\ 

$21 \times 21$ & $28718 \times 27839 (5)$ & 14.991 & $28718 \times 27839 (5)$  & 8.226 & 924 \\ 

$22 \times 22$ & $34777 \times 35383(5)$ & 28.947 & $34777 \times 35383(5)$  & 28.840 & 0  \\

HFE\_25\_96 & $63341 \times 68298(5)$ & 131.806 & $14271\times 15222(4)$ & 3.418 & 300 \\ 

HFE\_30\_96 & $141708 \times 174303(5)$  & 1504.289  & $22745\times 31861(4)$ & 15.168 & 360 \\ 

HFE\_35\_96 & $248520\times 383915(5)$  & $>1h$  & $33929\times 59410(4)$ & 57.988 & 420 \\ 

\hline
\end{tabular}\\
\caption{\small GVW vs M-GVW}
\end{table}
}

From this table, we can find that the maximal  size of the matrix generated during the computations are  exactly the same  for Courtois' examples except $21 \times 21$,
 and the corresponding computing time are also almost the same. This is  because we can not  find mutant polynomials with degree lower than 4 for  these examples.
 It is a little bit surprised that we find 924 mutant pairs with degree smaller than 4 in the example $21 \times 21$.  For the HFE examples, M-GVW performs much better than GVW  because many mutant polynomials have been found in M-GVW and the maximal size of the matrix in M-GVW become much smaller than that in GVW, which  leads that M-GVW cost less computing time.

We also test the same examples as in Table 2  for  M-GVW and some intrinsic implementations on public softwares, including Gr\"obner basis functions on Maple (version 17, setting ``method = fgb''), Singular (version 3-1-6),
and Magma (version 2.12-16)\footnote{Magma 2.12-16 is an old version, and we are trying to buy the latest one. }, and the computing times in seconds are listed in Table 3.

{
\begin{table}[H]\centering

\medskip
\begin{tabular}{c c c c c}
\hline
Exam. & Maple & Singular & Magma & M-GVW\\ \hline

$16 \times 16$ & 4.088 & 5.210 & 0.484 &  0.543 \\

$17 \times 17$ & 9.891 & 12.886 & 0.874 & 0.895\\

$18 \times 18$ & 22.340 & 31.590 & 1.513 &  1.588 \\

$19 \times 19$ & 48.314 & 84.771 & 2.792 & 2.728\\

$20 \times 20$ & 107.064 & 265.325 & 5.226 & 4.664   \\

$21 \times 21$ & 218.479 & 724.886 & 10.468 & 8.226\\

$22 \times 22$ & 839.067 & $>1h$  & 37.144 & 28.840\\

HFE\_25\_96 & 121.681 & $>1h$ &  7.675 & 3.418\\

HFE\_30\_96 & 619.745 & $>1h$ & 29.172  & 15.168  \\

HFE\_35\_96 & 2229.239 & $>1h$ &  out mem. &  57.988 \\
\hline
\end{tabular}\\
\caption{\small  Maple, Singular and Magma vs M-GVW}
\end{table}
}

From the above table, we can see that, due to the efficiency of routines from M4RI, our implementation of M-GVW is more efficient than some of functions from existing public softwares. However, since the matrices in large polynomial systems become quite sparse, our implementation may not perform very good for large systems at present.



\section{Conclusions} \label{sec_conclusion}
In this paper, we present M-GVW to avoid criteria rejecting J-pairs with lower degrees. M-GVW is exactly the same as GVW when input systems are homogeneous,
but have a better performance when input systems are inhomogeneous.
Due to the efficient routines from M4RI, we also give an efficient implementation of M-GVW using linear algebra over boolean polynomial rings. We think our implementation can be optimized further, and we will try to use sparse linear algebra to improve the performance of M-GVW in the future.

\end{document}